\documentclass[11pt,a4paper]{article}
\usepackage{amssymb}
\usepackage{amsmath}
\usepackage{amsfonts}
\usepackage{float}
\usepackage[pdftex]{graphicx}
\graphicspath{{./graphics/}}
\DeclareGraphicsExtensions{.pdf,.png,.jpg}
\usepackage{amscd}
\usepackage{a4wide}
\usepackage{microtype}
\usepackage{slashed}
\usepackage{fancyhdr}
\usepackage{amsthm}
\usepackage{mathrsfs}
\usepackage{hyperref}
\usepackage{color} 
\usepackage{pstricks}
\usepackage{epstopdf}

\definecolor{darkred}{rgb}{0.8,0.1,0.1}
\hypersetup{
     colorlinks=true,         
     linkcolor=darkred,
     citecolor=blue,
}

\usepackage{comment}
\usepackage[margin=2.7cm]{geometry}
\usepackage[all,cmtip]{xy}

\usepackage{marvosym}

\newtheorem{theorem}{\rmfamily\bfseries{Theorem}}[section] 
\newtheorem{corollary}[theorem]{\rmfamily\bfseries{Corollary}} 

\newtheorem{lemma}[theorem]{\rmfamily\bfseries{Lemma}} 
\newtheorem{proposition}[theorem]{\rmfamily\bfseries{Proposition}}

\newtheorem{definition}[theorem]{\rmfamily\bfseries{Definition}}

\theoremstyle{remark}

\newtheorem{remark}{Remark}

\def\un{1\kern-3pt \rm I}

\numberwithin{equation}{section}

\def\1{\mathbbm{1}}

\newcommand{\oR}{{\mathbb R}}

\newcommand{\oZ}{{\mathbb Z}}

\title{\bf{On the two-dimensional Schr\"odinger operator \\ with an attractive potential of the \\ Bessel-Macdonald type}}

\author{
      Wellisson B. De Lima,$^\dag$ Oswaldo M. Del Cima,$^\dag$ Daniel H.T. Franco$^\dag$ and Bruno C. Neves$^{\ddag}$\vspace{4mm}\\
      $^\dag$ Grupo de F\'\i sica-Matem\'atica e Teoria Qu\^antica dos Campos,\\
      Universidade Federal de Vi\c cosa, Departamento de F\'\i sica,\\
      Av. Peter Henry Rolfs s/n, Campus Universit\'ario,\\
      Vi\c cosa, MG, Brasil, CEP: 36570-900.\vspace{4mm}\\
      $^\ddag$ Departamento de Astrof\'\i sica, Cosmologia e Intera\c c\~oes Fundamentais,\\
      Centro Brasileiro de Pesquisas F\'\i sicas,\\
      Rua Dr. Xavier Sigaud 150, Urca,\\
      Rio de Janeiro, RJ, Brasil, CEP: 22290-180.\vspace{4mm}\\
{\footnotesize e-mail: \texttt{wellisson.lima@ufv.br, oswaldo.delcima@ufv.br, daniel.franco@ufv.br, bruno.lqg@gmail.com}}
}

\date{\today}

\begin{document}

\maketitle

\begin{abstract}
We analyze the Schr\"odinger operator in two-dimensions with an attractive potential given by a Bessel-Macdonald
function. This operator is derived in the non-relativistic approximation of planar quantum electrodynamics (${\rm QED}_3$) models as a framework for evaluation of two-quasiparticle scattering potentials. The analysis is motivated keeping 
in mind the fact that parity-preserving ${\rm QED}_3$ models can provide a possible explanation for the behavior of
superconductors. Initially, we study the self-adjointness and spectral properties of the Schr\"odinger operator modeling the non-relativistic approximation of these ${\rm QED}_3$ models. Then, by using {\em Set\^o-type estimates}, an estimate is derived of the number of two-particle bound states which depends directly on the value of the effective coupling constant, $C$, for {\em any} value of the angular momentum. In fact, this result in connection with the condition that guarantees the self-adjointness of the Schr\"odinger operator shows that there can always be a large number of two-quasiparticle bound states in planar quantum electrodynamics models. In particular, we show the existence of an isolated two-quasiparticle bound state if the effective coupling constant $C \in (0,2)$ in case of zero angular momentum.
To the best of our knowledge, this result has not yet been addressed in the literature. Additionally, we obtain an explicit estimate for the energy gap of two-quasiparticle bound states which might be applied to high-$T_c$ $s$-wave Cooper-type superconductors as well as to $s$-wave electron-polaron--electron-polaron bound states (bipolarons) in mass-gap graphene systems.
\end{abstract}

\section{\bf Introduction}
\label{Sec1}
\hspace*{\parindent}
The quantum electrodynamics in three space-time dimensions (QED$_3$) has been drawn attention, 
since the works by Schonfeld, Deser, Jackiw and Templeton~\cite{schonfeld,deser-jackiw-templeton}, 
as a potential theoretical framework to be applied to quasi-planar condensed matter systems~\cite{deandrade-delcima-helayel}, namely high-$T_{\rm c}$ superconductors~\cite{high-Tc,christiansen-delcima-ferreira-helayel}, quantum Hall effect~\cite{quantum-hall-effect}, topological insulators~\cite{topological-insulators}, topological superconductors~\cite{topological-superconductors} and graphene~\cite{graphene,graphene1,graphene2}. Thenceforth, planar quantum electrodynamics models have been studied in many physical configurations: small (perturbative) and large (non perturbative) gauge transformations, abelian and non-abelian gauge groups, fermions families, even or odd under parity, compact space-times, space-times with boundaries, curved space-times, discrete (lattice) space-times, external fields 
and finite temperatures. In condensed matter systems, quasiparticles usually stem from two-particle (Cooper pairs), 
particle-quasiparticle (excitons) or two-quasiparticle (bipolarons) non relativistic bound states. Bearing in mind 
these issues together with the fact that there are QED$_3$ models in which, fermion-fermion, fermion-antifermion 
or antifermion-antifermion scattering potentials -- mediated by massive\footnote{Otherwise, if the mediated quanta 
were massless, the interaction potential would be a logarithm-type (confining) potential~\cite{maris}.} scalars or vector mesons -- can be attractive and of $K_0$-type (a Bessel-Macdonald function)~\cite{christiansen-delcima-ferreira-helayel,Wado1}, we propose to study the Schr\"odinger equation in three space-time dimensions by using the modified Bessel function of the second kind, $K_0$, as the interaction radial potential. It should be stressed that a $K_0$-type two-quanta scattering potential, $K_0(r/\lambda)$, is a two-dimensional nonconfining interaction possessing a length scale ($\lambda$), namely, the Compton wavelength associated to the mediated quantum field, therefore it might be a strong candidate for describing two-quanta bound states in condensed matter planar systems. 

In this work, there are essentially
two main results: firstly, for the sake of completeness, we prove that the considered particular potential belongs to 
a general class of potentials for which self-adjointness is guaranteed. In other words, we prove the ``smallness'' 
of this potential relative to the free hamiltonian operator $H_0$, in the sense of Kato, implying the self-adjointness 
of the hamiltonian operator $H=H_0+V$, where $V(r)=-\alpha K_0(\beta r)$ is the attractive two-particle 
scattering potential, with $\boldsymbol{\mathfrak{Dom}}(H)=\boldsymbol{\mathfrak{Dom}}(H_0)$. 
We also get information about the discrete and essential spectra of the Schr\"odinger operator modeling the
non-relativistic approximation of the model. Posteriorly, Bargmann-type bounds on the number of negative eigenvalues are obtained. More specifically, by using {\em Set\^o-type estimates}~\cite{Seto}, we obtain an upper limit for the number of two-quantum bound states for any value of the angular momentum. Consequently, this result in conjunction with the condition that assures the self-adjointness of the hamiltonian implies that there can ever be a non vanishing number of two-quantum bound states for any $K_0$-type attractive interaction potential. To the best of our knowledge, 
this result has not yet been addressed in the literature. This corroborates the well-known fact that in two space 
dimensions arbitrary weak potentials always possess at least one bound state~\cite{Yang,Chadan}. Indeed, we show the existence of an isolated two-quasiparticle bound state if the effective coupling constant $C \in (0,2)$ in case of zero angular momentum. To be more specific, the latter result might be relevant to describe high-$T_c$ $s$-wave Cooper-type pairing superconductors~\cite{deandrade-delcima-helayel,high-Tc,christiansen-delcima-ferreira-helayel} as well as $s$-wave electron-polaron--electron-polaron bound states in mass-gap graphene systems~\cite{graphene,graphene1,graphene2}.

\section{\bf Non-relativistic planar quantum electrodynamics}
\label{Sec2}
\hspace*{\parindent}
The Schr\"odinger operator modeling the non-relativistic approximation of the parity-preserving ${\rm QED}_3$
is~\cite{deandrade-delcima-helayel,christiansen-delcima-ferreira-helayel,graphene2,Wado1}
\begin{align}
H=H_0+V=-\frac{\hbar^2}{2\mu} \Delta(\boldsymbol{x}) - \alpha K_0(\beta \|\boldsymbol{x}\|)\,\,,
\label{HOper}
\end{align}
where $\mu$ is the reduced mass and $\alpha$ is the coupling parameter taken to be, without the loss of generality,
non-negative. The constants $\alpha$ and $\beta$ shall depend on some model parameters, like coupling constants,
characteristic lengths, mass parameters or vacuum expectation value of a scalar 
field~\cite{deandrade-delcima-helayel,christiansen-delcima-ferreira-helayel,graphene2,Wado1}.

\begin{remark}
Throughout the article, we will {\em not} use atomic units $\hbar=\alpha=2\mu=1$, as is common in the literature.
\end{remark} 

Taking into account that the potential $K_0(\beta \|\boldsymbol{x}\|)$ only depends on $ \|\boldsymbol{x}\|$, the
distance from origin, in order to estimate the number of two-particle bound states (in Section \ref{Sec4}) we shall introduce
polar coordinates $(r=\|{\boldsymbol{x}\|,\theta})$, so that
\[
L_2(\oR^2)=L_2((0,\infty); r dr) \otimes L_2(S^1,d\theta)\,\,,
\]
where $S^1$ is the usual unit circle in $\oR^2$. Let $D$ be the set of all functions that are linear combinations of 
products $\Psi(r)\Theta(\theta)$ with $\Psi \in L_2((0,\infty); r dr)$ and $\Theta \in L_2(S^1,d\Theta)$.
Then, by the separation of variables for the two-dimensional Schr\"odinger 
equation associated to the both two-particle pairing states
\begin{align*}
\left[-\frac{\hbar^2}{2\mu} \Delta - \alpha K_0(\beta \|\boldsymbol{x}\|)\right] \Psi(\boldsymbol{x})
=E \Psi(\boldsymbol{x})
\quad \boldsymbol{x} \in \oR^2 \qquad (\text{$E$: energy})\,\,,
\end{align*}
into the radial part and angular part, the radial part takes the form
\begin{align}
\left[-\frac{\hbar^2}{2\mu} \left(\frac{d^2}{dr^2}+\frac{1}{r}\frac{d}{dr}\right)
+\frac{\hbar^2 m^2}{2\mu r^2}-\alpha K_0(\beta r)\right] \Psi(r)=E \Psi(r)\,\,,
\label{EqSch}
\end{align}
while the angular part takes the form
\begin{align}
\frac{d^2\Theta(\theta)}{d\theta^2}=-m^2 \Theta(\theta)\,\,.
\label{EqSch1}
\end{align}
The operator $d^2/d\theta^2$, with domain $C_0^\infty(S^1)$, is essentially self-adjoint~\cite{EssOp}. Its eigenvectors 
$\Theta_m(\theta)=(2\pi)^{-1/2} e^{im\theta}$, with $m \in \oZ$, constitute an orthonormal basis of $L_2(S^1,d\theta)$.

Let $\Omega_m$ denote the subspace spanned by $\Theta_m$ and $L_m=L_2((0,\infty); r dr) \otimes \Omega_m$.
Then,
\[
L_2(\oR^2)=\bigoplus_{m \in \oZ} L_m\,\,.
\]
If ${\un}_m$ is the identity operator on $\Omega_m$, the restriction of hamiltonian operator
$H$ to $D_m=D \cap L_m$ is given by $H \bigr|_{D_m}=H_m \otimes {\un}_m$, with
\[
H_m=-\frac{\hbar^2}{2\mu} \left(\frac{d^2}{dr^2}+\frac{1}{r}\frac{d}{dr}\right)
+\frac{\hbar^2 m^2}{2\mu r^2}-\alpha K_0(\beta r)\,\,.
\]
In terms of a function $\Phi$ defined by
\begin{align}
\Phi(r)=r^{1/2} \Psi(r)\,\,,
\label{WaveFunc}
\end{align}
$H_m$ is expressed as
\begin{align}
H_m=-\frac{\hbar^2}{2\mu} \frac{d^2}{dr^2}
+\left(\frac{\hbar^2 (m^2-1/4)}{2\mu r^2}-\alpha K_0(\beta r)\right)\,\,.
\label{EqSch2}
\end{align}
This reduces the problem of the hamiltonian operator, $H$, in two-dimensions in a 
problem of the hamiltonian operator, $H_m$, in one-dimension, with an effective potential
given by:
\[
V_{\rm eff}(r) \overset{\rm def.}{=} \frac{\hbar^2 (m^2-1/4)}{2\mu r^2}-\alpha K_0(\beta r)\,\,.
\]

\begin{remark}
According to Ref.\cite[Lemma 2.1]{Seto}, the wave function (\ref{WaveFunc}) of a bound state with negative 
energy level behaves as
\begin{align*} 
\Phi(r)=
\begin{cases}
O(r^{m+\frac{1}{2}}) \quad &\text{for} \quad r \to 0 \\[3mm]
O(e^{-kr}) \quad &\text{for} \quad r \to \infty; k=(-2\mu E)^{\frac{1}{2}} \hbar^{-1}
\end{cases}\,\,.
\end{align*}
Therefore such a function $\Phi$ belongs to ${\mathscr H}= L_2((0,\infty); |V(r)| dr)$, as
well as to $L_2((0,\infty); dr)$.
\end{remark}

\begin{remark}
From expression $V(r)=-\alpha K_0(\beta r)$ we see that $\alpha$ has energy dimension and 
gives us an energy scale for the interaction among the two particles. In turn, the parameter $\beta$ 
has inverse length dimension, thus fixing a length scale, an interaction range, which is related to 
the mass of the boson-mediated quantum ($M_b$) exchanged during the two particle 
scattering~\cite{christiansen-delcima-ferreira-helayel,graphene2,Wado1}. This can be verified if we consider the 
Compton wavelength of the boson-mediated field, $\lambda_c=2\pi \hbar (M_b c)^{-1}$, hence 
$\beta=\lambda_c^{-1}=M_b c (2\pi \hbar)^{-1}$.  Also, if we take the constant, $\hbar^2 \beta^2 (2\mu)^{-1}$,
which has energy dimension, together with the relation among $\beta$ and $M_b$, an energy scale is fixed
as well. Thus, we introduce a dimensionless constant $C=2\mu \alpha (\hbar\beta)^{-2}$ that gives a notion
of how strong is the two quanta interaction ($\alpha$) when compared to the energy of the boson-mediated
quantum ($M_b c^2$). By taking into account this analysis, we rewrite the effective potential in a most convenient 
way:
\begin{align}
v_{\rm eff}(s)=\frac{(m^2-1/4)}{s^2}-CK_0(s)\,\,.
\label{EffecPot}
\end{align}
where we define
\[
s=\beta r\,\,,
\quad
C=\frac{2\mu \alpha}{\hbar^2 \beta^2}\,\,.
\]
Hence,
\[
v_{\rm eff}(s)=\frac{2\mu}{\hbar^2 \beta^2} V_{\rm eff}(s)\,\,.
\]
\label{Nota2}
\end{remark}

\section{\bf Self-adjointness of the Schr\"odinger operator and all that}
\label{Sec3}
\hspace*{\parindent}
In this section, the preservation of self-adjointness of the free particle hamiltonian $H_0$ under 
small symmetric perturbations is considered. In particular, the application of Kato-Rellich Theorem to the
hamiltonian (\ref{HOper}) is discussed. This theorem is a cornerstone in the theory of self-adjointness
for hamiltonian operators $H_0+V$. As a starting point, let us remember the following (see, for example, Reed-Simon~\cite{RS})

\begin{definition}[Kato's Criterion]
Suppose $A,B$ are two densely defined linear operators in ${\mathscr H}$.\footnote{Throughout this article
${\mathscr H}$ is assumed to be a complex Hilbert space.} $B$ is called a {\bf Kato perturbation} of $A$ 
if, and only if, $\boldsymbol{\mathfrak{Dom}}(A) \subset \boldsymbol{\mathfrak{Dom}}(B)$ and there are 
non-negative real numbers $0 \leqslant a < 1$ and $b(a) \in (0,\infty)$ such that
\begin{align}
\|B \Psi\| \leqslant a \|A \Psi\| + b \|\Psi\|\,\,,
\quad \forall\, \Psi \in \boldsymbol{\mathfrak{Dom}}(A)\,\,.
\label{KC}
\end{align}
\label{CK}
In this case, $B$ is said to be $A$-bounded.
\end{definition}

Generally, in the Definition \ref{CK}, $b$ must be chosen larger as $a$ is chosen
smaller. In other words, we can increase $b$ if we decrease $a$, but we cannot take $a=0$ 
for any finite $b$ unless $B$ is a bounded operator. For this reason, we have $b=b(a)$.
The infimum (greatest lower bound) of the possible $a$ is called the {\em relative bound} of $B$ with respect to $A$.

The notion of a Kato perturbation is very effective in solving the problem of self-adjointness of the sum,
under natural restrictions. The next result is the celebrated (see Reed-Simon~\cite[Theorem X.12]{RS})

\begin{theorem}[Kato-Rellich Theorem] 
Suppose $A$ is a self-adjoint and $B$ is a symmetric operator in ${\mathscr H}$. If $B$ 
is a Kato perturbation of $A$, then the sum $A+B$ is self-adjoint on the domain
$\boldsymbol{\mathfrak{Dom}}(A)$ and essentially self-adjoint on any core of $A$.
Moreover, if $A$ is bounded below by $M$, then $A+B$ is bounded from below by 
$M-\max\bigl\{b/(1-a),a|M|+b\bigr\}$, where $a$ and $b$ are given by (\ref{KC}).
\label{KRTheo}
\end{theorem}

Now the Kato-Rellich Theorem will be applied to small perturbations of the free particle hamiltonian, $H_0$. The
domain
\begin{align}
\boldsymbol{\mathfrak{Dom}}(H_0)&=\left\{ \Psi \in L_2(\oR^n) \mid \text{$-\frac{\hbar^2}{2\mu}\Delta \Psi(\boldsymbol{x})
\in L_2(\oR^n)$ in the sense of distributions} \right\}\,\,,
\label{Dommax}
\end{align}
is discussed in details by Reed-Simon~\cite[Theorem IX.27]{RS}. 

\begin{remark}
The domain (\ref{Dommax}) is equivalent to the condition~\cite{FT}
\begin{align*}
\boldsymbol{\mathfrak{Dom}}(H_0)&=\left\{ \Psi \in L_2(\oR^n) \mid \text{$\frac{1}{2\mu} \|\boldsymbol{k}\|^2 \widehat{\Psi}(\boldsymbol{k})
\in L_2(\oR^n)$ in the sense of distributions}  \right\}\,\,.
\end{align*}
Moreover, $\boldsymbol{\mathfrak{Dom}}(H_0)={\cal H}_{(2)}(\oR^n)$, which is exactly the Sobolev 
space of order $2$. And more, as the free hamiltonian, $H_0$ is self-adjoint, then the powers, $H^m_{_0}$, 
of the free hamiltonian are also self-adjoint, since $H^m_{_0}={\mathscr F}^{-1} \|\boldsymbol{k}\|^m {\mathscr F}$. Therefore, 
the domain of $H^m_{_0}$ is the Sobolev space of order $2m$. This immediately implies that a vector $\Psi \in L_2(\oR^n)$
it is at $C^\infty(H_0)=\bigcap_{m=1}^\infty \boldsymbol{\mathfrak{Dom}}(H^m_{_0})$ if, and only if,
$\Psi \in C^\infty(\oR^n)$ and $D^\kappa \Psi \in L_2(\oR^n)$ for all $\kappa$.
\end{remark}

Let us see how the Kato's Criterion and the Kato-Rellich Theorem allows us to establish the self-adjointness of 
the Schr\"odinger operator  (\ref{HOper}).  

\begin{theorem} The Schr\"odinger operator (\ref{HOper}) in $L_2(\oR^2)$ is self-adjoint on the domain 
$\boldsymbol{\mathfrak{Dom}}(H_0)$.
\label{pertubHo}
\end{theorem}

Naturally, implicit in the statement of the Theorem \ref{pertubHo} is that $\boldsymbol{\mathfrak{Dom}}(V)$
contains $\boldsymbol{\mathfrak{Dom}}(H_0)$. Besides that, the meaning of the operator $V$ is clear: it is a real
multiplicative operator, and its domain $\boldsymbol{\mathfrak{Dom}}(V)$ consists of all $\Psi \in {\mathscr H}$ as given 
in~\cite[Proposition 9.30]{BHall}. Thus defined, $V$ is obviously self-adjoint.

The proof of the above Theorem \ref{pertubHo} requires the following

\begin{lemma}
All $\Psi \in \boldsymbol{\mathfrak{Dom}}(H_0) \subset L_2(\oR^2)$ are bounded by
\[
\|\Psi\|_\infty \leqslant \frac{\mu^{1/2}}{(2\pi)^{3/2} \hbar \lambda} \left(\lambda^2\|{\Psi}\|_2 + \|H_0{\Psi}\|_2 \right)\,\,.
\]
\label{lemmaAB}
\end{lemma}

\begin{proof}
Following the same reasoning taken from Ref.\cite[Lemma 6.2.1]{Oliv}, we take an arbitrary constant $\lambda > 0$ and
consider that for $n \leqslant 3$, for every $\Psi \in \boldsymbol{\mathfrak{Dom}}(H_0) \subset L_2(\oR^n)$,
the function $\boldsymbol{k} \mapsto (\lambda^2+(2\mu)^{-1} \|\boldsymbol{k}\|^2)^{-1} \in L_2(\oR^n)$. Moreover 
$L_2(\oR^n) \ni (\lambda^2+(2\mu)^{-1} \|\boldsymbol{k}\|^2)\widehat{\Psi}(\boldsymbol{k})={\mathscr F}(\lambda^2\Psi+H_0\Psi)$~\cite{FT}.
Therefore, by the H\"older inequality, 
\[
(\lambda^2+(2\mu)^{-1} \|\boldsymbol{k}\|^2)^{-1}
(\lambda^2+(2\mu)^{-1}\|\boldsymbol{k}\|^2)\widehat{\Psi}(\boldsymbol{k}) \in L_1(\oR^2)\,\,,  
\]
and
\begin{align*}
\|\widehat{\Psi}\|_1&=\int_{\oR^2} \frac{d^2\boldsymbol{k}}{\hbar^2}\,\,
(\lambda^2+(2\mu)^{-1} \|\boldsymbol{k}\|^2)^{-1}
(\lambda^2+(2\mu)^{-1} \|\boldsymbol{k}\|^2)|\widehat{\Psi}(\boldsymbol{k})| \\[3mm]
&\leqslant \left(\int_{\oR^2} \frac{d^2\boldsymbol{k}}{\hbar^2}\,\,(\lambda^2+(2\mu)^{-1} \|\boldsymbol{k}\|^2)^{-2}\right)^{1/2} 
\left(\int_{\oR^2} \frac{d^2\boldsymbol{k}}{\hbar^2}\,\,(\lambda^2+(2\mu)^{-1} 
\|\boldsymbol{k}\|^2)^2|\widehat{\Psi}(\boldsymbol{k})|^2\right)^{1/2}\!\!\!\,.
\end{align*}

The first integral can be easily calculated. Indeed, by using the Table of Integrals of Gradshteyn-Ryzhik~\cite[{\bf 3.241}, $4.^{11}$, p.322]{GR}, we obtain
\[
\int_{\oR^2} \frac{d^2\boldsymbol{k}}{\hbar^2}\,\,(\lambda^2+(2\mu)^{-1} \|\boldsymbol{k}\|^2)^{-2}
=\frac{8 \pi \mu^2}{\hbar^2} \underbrace{\int_{0}^\infty dk\,\,k (2\mu\lambda^2+k^2)^{-2}}_{1/(4\mu\lambda^2)}
=\frac{2\pi \mu}{\hbar^2 \lambda^2}\,\,.
\]

Now, using Minkowski inequality, we have
\begin{align*}
\|\widehat{\Psi}\|_1
&\leqslant \frac{(2\pi \mu)^{1/2}}{\hbar \lambda} \left\|(\lambda^2+(2\mu)^{-1}\|\boldsymbol{k}\|^2)\widehat{\Psi}\right\|_2 \\[3mm]
&\leqslant \frac{(2\pi \mu)^{1/2}}{\hbar \lambda} \left(\lambda^2\|\widehat{\Psi}\|_2
+(2\mu)^{-1} \left\|\|\boldsymbol{k}\|^2\widehat{\Psi}\right\|_2 \right)\,\,.
\end{align*}
On the other hand, by the inverse Fourier transform~\cite{FT},
\begin{align*}
\Psi(\boldsymbol{x})=\frac{1}{(2\pi \hbar)^2} \int_{\oR^2}
d^2\boldsymbol{k}\,\,e^{-i \hbar^{-1}\boldsymbol{k} \cdot \boldsymbol{x}}\,\widehat{\Psi}(\boldsymbol{k})\,\,,
\end{align*}
we obtain the estimate $\|\Psi\|_\infty \leqslant (2\pi)^{-2} \|\widehat{\Psi}\|_1$ well known. This implies that,
\begin{align*}
\|\Psi\|_\infty
&\leqslant \frac{\mu^{1/2}}{(2\pi)^{3/2} \hbar \lambda}\left(\lambda^2\|\widehat{\Psi}\|_2 
+(2\mu)^{-1}\left \|\|\boldsymbol{k}\|^2\widehat{\Psi}\right\|_2\right) \\[3mm]
&=\frac{\mu^{1/2}}{(2\pi)^{3/2} \hbar \lambda} \left(\lambda^2\|{\Psi}\|_2 + \|H_0{\Psi}\|_2 \right)\,\,,
\end{align*}
since the Fourier transform is a unitary operator. This completes the prove.
\end{proof}

\begin{proof}[Proof of Theorem \ref{pertubHo}] 
Firstly, we shall show that $V(\boldsymbol{x})=-\alpha K_0(\beta \|\boldsymbol{x}\|) \in L_2(\oR^2)$.
By using the Table of Integrals of Gradshteyn-Ryzhik~\cite[{\bf 6.521}, $6.^*$, p.665]{GR}, we obtain
\begin{align*}
\|V\|_2=\left(\alpha^2 \int_{\oR^2} d^2\boldsymbol{x}\,\,K^2_{0}(\beta \|\boldsymbol{x}\|)\right)^{1/2}
=\left(2 \pi \alpha^2 \underbrace{\int_{0}^\infty dr\,\,r K^2_{0}(\beta r)}_{1/2\beta^2}\right)^{1/2}
=\frac{\pi^{1/2} \alpha}{\beta}\,\,.
\end{align*}
Therefore, $V(\boldsymbol{x})=-\alpha K_0(\beta \|\boldsymbol{x}\|) \in L_2(\oR^2)$. Obviously, as a multiplication
operator, V is closed on its domain of definition; this implies that if $\Psi \in L_2(\oR^2)$, then $V\Psi \in L_2(\oR^2)$.

Now, again, using the H\"older inequality, it follows that
$\|V\Psi\|_2 \leqslant \|V\|_2 \|\Psi\|_\infty$. Thus, by Lemma \ref{lemmaAB}, we obtain
\begin{align*}
\|V\Psi\|_2 &\leqslant \|V\|_2 \frac{\mu^{1/2}}{(2\pi)^{3/2} \hbar \lambda} \left(\lambda^2\|{\Psi}\|_2 + \|H_0{\Psi}\|_2 \right) \\[3mm]
&=\frac{\mu^{1/2} \alpha}{2^{3/2} \pi \hbar \lambda \beta} \left(\lambda^2\|{\Psi}\|_2 + \|H_0{\Psi}\|_2 \right)\,\,.
\end{align*}
We define 
\begin{align}
a(\lambda)=\frac{\mu^{1/2} \alpha}{2^{3/2} \pi \hbar \lambda \beta}
\quad \text{and} \quad 
b(\lambda)=\frac{\mu^{1/2} \alpha \lambda}{2^{3/2} \pi \hbar \beta}\,\,.
\label{ab}
\end{align}
Since $\lambda$ is an arbitrary positive constant, just assume that ${\mu^{1/2} \alpha}/({2^{3/2} \pi \hbar \beta}) < \lambda$
for the factor $a(\lambda)$ to be smaller than $1$. The latter proves that the potential 
$V(\boldsymbol{x})=-\alpha K_0(\beta \|\boldsymbol{x}\|)$ is $H_0$-bounded, so that Theorem \ref{KRTheo}
applies and proves the self-adjointness of the Schr\"odinger operator (\ref{HOper}).
\end{proof}

\begin{remark}
At this point, remember that the physical interpretation of the wavefunction is that $d^2\boldsymbol{x}\,\,|\Psi(\boldsymbol{x})|^2$
gives the probability of finding the quantum particle in a region $d^2\boldsymbol{x}$ around the position $\boldsymbol{x}$.
Probability is a dimensionless quantity. Hence, $|\Psi(\boldsymbol{x})|^2$ must have dimension of inverse area $L^{-2}$ and
$\Psi$ has dimension $L^{-1}$. Similarly, the physical interpretation of the wavefunction in momenta space is that
$\hbar^{-2} d^2\boldsymbol{k}\,\,|\widehat{\Psi}(\boldsymbol{k})|^2$ gives the probability of finding the quantum particle
in a region $\hbar^{-2} d^2\boldsymbol{k}$ around the momentum $\boldsymbol{k}$. Hence $|\widehat{\Psi}(\boldsymbol{k})|^2$
must have dimension of square length $L^2$ and $\widehat{\Psi}$ has dimension $L$. The dimensions of all quantities involved
in the proofs of Lemma \ref{lemmaAB} and Theorem \ref{pertubHo} are collected in Table \ref{table1}.

\begin{table}
\begin{center}
\begin{tabular}{|c||c|c|c|c|c|c|c|c|c|c|c|c|}
\hline
&$\widehat{\Psi}$ &$|\widehat{\Psi}|^2$ &$\hbar^{-2} d^2\boldsymbol{k}$ &$\lambda^2$ &$\alpha$ &$\beta$ &$\mu$ &$\|\widehat{\Psi}\|_1$ 
&$\|\Psi\|_\infty$ &$\|V\|_2$ &$a(\lambda)$ &$b(\lambda)$\\
\hline\hline
${\rm Dimensions}$ &$L$ &$L^2$ &$L^{-2}$ &$E$ &$E$ &$L^{-1}$ &$M$ &$L^{-1}$ &$L^{-1}$ &$E L$ &$\varnothing$ &$E$ \\
\hline
\end{tabular}
\end{center}
\caption[]{Dimension of $\widehat{\Psi}$ and all that.}
\label{table1}
\end{table}
\end{remark}

In applications it is often very important to determine the lowest point of the spectrum of
a self-adjoint operator. This problem makes sense only if the operator is bounded from
below, since otherwise the spectrum extends to $-\infty$. The boundedness from below
of the Schr\"odinger operator with $K_0$-potential is analyzed below.

\begin{theorem}
The Schr\"odinger operator (\ref{HOper}) is bounded from below by $-C\alpha/4\pi$,
where $C=2\mu \alpha (\hbar\beta)^{-2}$ is a dimensionless constant, while $\alpha$ has energy
dimension and gives us an energy scale for the interaction among the two particles.
\label{TheoBFB}
\end{theorem}

This assertion can be easily proved with the help of the following result, with the proof being obtained
directly from Ref.\cite[Theorem 9.37]{BHall} and from Ref.\cite[Theorem 18.6.1]{Shlomo}.

\begin{lemma}
Let $V$ be a Kato potential. Then $H=H_0+V$ is bounded from below by $-b(\lambda)/\bigl(1-a(\lambda)\bigr)$.
\label{Hall}
\end{lemma}

\begin{proof}
In verifying the lemma it is sufficient to do for all $\Psi$ belonging to a core of $H_0$. Thus, without loss
of generality, we can assume that $\Psi \in C_0^\infty(\oR^n)$. According to Theorem 7.6 in Ref.\cite{HS},
the spectrum of the operator $H_0$ on $\boldsymbol{\mathfrak{Dom}}(H_0)$ is $\sigma(H_0)=[0,\infty)$. This
means that
\[
\inf_{\substack{\Psi \in \boldsymbol{\mathfrak{Dom}}(H_0) \\ \|\Psi\|=1}} \langle \Psi,H_0\Psi \rangle
=\inf \sigma(H_0)=0\,\,.
\]
Thus, $H_0$ is bounded from below by zero and by Theorem \ref{KRTheo} and Theorem \ref{pertubHo} we get that $H_0+V$
is bounded from below by $-b(\lambda)/\bigl(1-a(\lambda)\bigr)$, where $a(\lambda)$ and $b(\lambda)$ are given by 
(\ref{ab}).
\end{proof}

\begin{proof}[Proof of the Theorem \ref{TheoBFB}]
According to Theorem \ref{pertubHo}, it is clear that the potential $K_0$ belongs to the class of potentials of Kato.
Therefore, according to Lemma \ref{Hall}, the spectrum of Schr\"odinger operator (\ref{HOper}) is bounded from below
by $-b(\lambda)/\bigl(1-a(\lambda)\bigr)$. Next, we shall define the function
\[
f(\lambda)=\frac{b(\lambda)}{1-a(\lambda)}
=\frac{A \lambda^2}{\lambda-A}\,\,,
\]
where
\[
A=\frac{\mu^{1/2} \alpha}{2^{3/2} \pi \hbar \beta}\,\,.
\]

To look for the extreme points of the above function (either maximum or minimum) it is first necessary to identify the
critical points of the function and then to check the sign of $f''$. The first and second derivatives of $f$ are read, respectively,
as follows
\begin{align}
f'(\lambda) = \frac{2A \lambda}{\lambda-A} - \frac{A \lambda^{2}}{(\lambda-A)^{2}}\,\,;
\label{1st-derivative}
\end{align} 
\begin{align}
f''(\lambda) = \frac{2A}{\lambda-A} - \frac{2A \lambda}{(\lambda-A)^{2}} - \frac{2A \lambda}{(\lambda-A)^{2}} 
+ \frac{2A \lambda^{2}}{(\lambda-A)^{3}}\,\,.
\label{2nd-derivative}
\end{align}
Now, since we have the first derivative, it follows from (\ref{1st-derivative}) that the critical point is
$\lambda_{_0}=2A$. With the critical point $\lambda_{_0}$ one can set $\lambda=2A$ on (\ref{2nd-derivative}) and one
easily obtain that $f''(\lambda_{_0}) > 0$. Thus, by the second derivative test $\lambda_{_0}=2A$ is indeed a minimum.
Consequently, we have
\[
f(\lambda)\Bigr|_{\lambda=\lambda_{_0}}=4A^2=\frac{C\alpha}{4\pi}\,\,.
\]
Hence, the Schr\"odinger operator (\ref{HOper}) is bounded from below by $-C\alpha/4\pi$.
\end{proof}

\begin{corollary}
The spectrum of the Schr\"odinger operator (\ref{HOper}) is situated in $[-C\alpha/4\pi,+\infty)$.
\end{corollary}

\begin{proof}
Naturally, according to Theorem \ref{TheoBFB} there is no spectrum below $-C\alpha/4\pi$. In other words, the boundedness
of $H$ from below means that there is an energy $E_{_0}=-C\alpha/4\pi$ (negative in the case considered here) such that
$\langle \Psi,H \Psi \rangle \geqslant E_{_0} \langle \Psi,\Psi \rangle$ for all $\Psi \in \boldsymbol{\mathfrak{Dom}}(H)$,
or equivalently that the entire spectrum lies in $E \geqslant  -C\alpha/4\pi$. It corresponds to the existence of bound states, 
with the state of lowest energy $-C\alpha/4\pi$ in spectrum of
$-\frac{\hbar^2}{2\mu} \Delta(\boldsymbol{x})-\alpha K_0(\beta \|\boldsymbol{x}\|)$. In conclusion, 
\[ 
\sigma\Bigl(-\frac{\hbar^2}{2\mu} \Delta(\boldsymbol{x})-\alpha K_0(\beta \|\boldsymbol{x}\|)\Bigr) 
=[-C\alpha/4\pi,+\infty)\,\,.
\]
\end{proof}

The result established in Theorem \ref{TheoBFB} guarantees that there are no eigenvalues
going off to $-\infty$. In the physics literature this is known as the ``{\em stability of matter.}'' 
Among other things, this result is particularly important as it ensures the stability of possible 
two-quasiparticle bound states (as discussed in Section \ref{Sec4}) at certain energy scales in
possible applications to two-dimensional materials such as high-$T_c$ superconductors, graphene
and topological insulators. 

Returning to the spectral properties of the Schr\"odinger operator (\ref{HOper}), we want information not
only about the spectrum of $H$, $\sigma(H)$, but also about the subsets $\sigma_{\rm disc}(H)$ and 
$\sigma_{\rm ess}(H)$. We remember that a useful decomposition of the spectrum of $H$ is the following:
$\sigma(H)=\sigma_{\rm disc}(H) \cup \sigma_{\rm ess}(H)$, $\sigma_{\rm disc}(H)$ being the {\em discrete spectrum}
of $H$, which is the set of all isolated eigenvalues of $H$ with finite algebraic multiplicity and $\sigma_{\rm ess}(H)$ being the 
{\em essential spectrum} of $H$, which is simply given by $\sigma_{\rm ess}(H)=\sigma(H) \setminus \sigma_{\rm disc}(H)$.
$\sigma_{\rm ess}(H)$ is stable under certain perturbations of $H$, more specifically, $\sigma_{\rm ess}(H)$
is stable under relatively compact perturbations of $H$. This result is the celebrated {\em Weyl's Invariance Theorem}.
The Weyl Theorem suffices to find $\sigma_{\rm ess}(H)$ for large class of two body Schr\"odinger operators
with their center of mass removed, as is our case.

If $V$ is real and $H_0$-bounded with relative $H_0$-bound $< 1$ and if $V(\boldsymbol{x}) \to 0$ as 
$\|\boldsymbol{x}\| \to \infty$, then $\sigma_{\rm ess}(H_0+V)=\sigma_{\rm ess}(H_0)=[0,\infty)$ on
$C_0^\infty(\oR^n)$ (see~\cite[Theorem 13.9]{HS}) (so $H$ can have only negative isolated eigenvalues, possibly
accumulating at 0). Consequently, in light of the fact that the potential $K_0$ is a real Kato potential, which tends
to zero at infinity (taking into account its asymptotic form), we have
\begin{proposition}
For the essential spectrum of the Schr\"odinger operator (\ref{HOper}) one has
\[
\sigma_{\rm ess}\Bigl(-\frac{\hbar^2}{2\mu} \Delta(\boldsymbol{x})-\alpha K_0(\beta \|\boldsymbol{x}\|)\Bigr)
=[0,\infty)\,\,.
\]
\label{Specprop1}
\end{proposition}

At or above the lower limit of the essential spectrum, any energy will affect the two-quantum bound state because $E > 0$.
The two-quantum pair will no longer be located, {\em i.e.}, the quasiparticles will move freely. Therefore, the Schr\"odinger
operator (\ref{HOper}) can only have discrete eigenvalues in $[-C\alpha/4\pi,0)$ (see~\cite[Proposition 12.8]{Konrad}).
In particular, under certain conditions exists exactly one two-quantum bound state, {\em i.e.}, only one eigenstate below the
bottom of the essential spectrum, as shown in the following

\begin{proposition}
Regarding the discrete part of the spectrum of the Schr\"odinger operator (\ref{HOper}), if $0 < C < 2$ for $m=0$, 
or if $C=2m$ for $m \geqslant 1$, where $m$ is the value of angular momentum, then
\[
\sigma_{\rm disc}\Bigl(-\frac{\hbar^2}{2\mu} \Delta(\boldsymbol{x})-\alpha K_0(\beta \|\boldsymbol{x}\|)\Bigr)
=\bigl\{-C\alpha/4\pi\bigr\}\,\,,
\]
with the multiplicity of the ground state being one.
\label{Specprop}
\end{proposition}

The proof of this Proposition will be given later, based on an estimate for the number of two-quasiparticle bound states.

\section{\bf Bounds on the number of two-quasiparticle bound states}
\label{Sec4}
\hspace*{\parindent}
In this section, we obtain an estimate for the number of two-quantum bound states. Firstly, note that for the case
$m=0$ (zero angular momentum), it is not difficult to see that the potential will be uniquely attractive, as we can
see through a direct inspection of Eq.(\ref{EffecPot}) (see Figure \ref{Well1}).

\begin{figure}[H]
\begin{center}
\scalebox{0.85} 
{
\begin{pspicture}
{\includegraphics{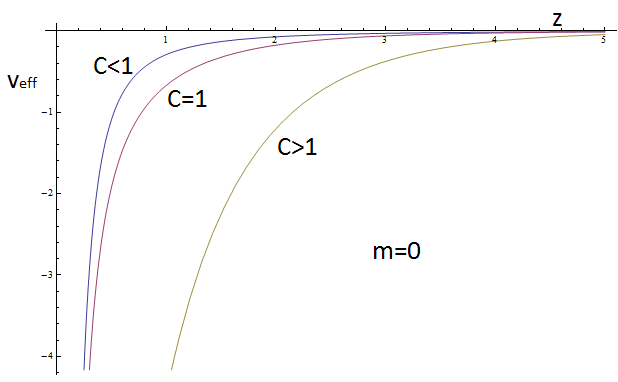}}
\end{pspicture} 
}
\end{center}
\caption{The effective potential for $m=0$ and some values of $C$.}
\label{Well1}
\end{figure}

The potential (\ref{EffecPot}) behaves, qualitatively, as the Coulomb potential in three dimensional space,
and if there are bound states in the model, these will probably appear for vanishing angular momentum ($m=0$)
state (or $s$-wave state as it is most known) since the potential is uniquely attractive. Note that if we keep the
parameter $C$ fixed and increase the angular momentum in a unit, the centrifugal term which is positive gives
a repulsive contribution to the effective potential and may even exceed the attractive term (see Figure \ref{Well2}).

\begin{figure}[H]
\begin{center}
\scalebox{0.85} 
{
\begin{pspicture}
{\includegraphics{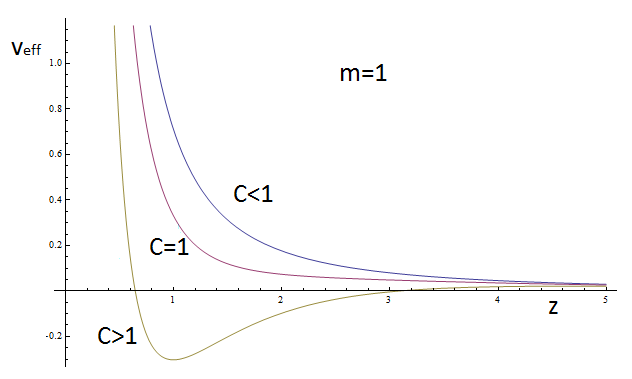}}
\end{pspicture} 
}
\end{center}
\caption{The effective potential for $m=1$ with the same values of $C$ as in Figure \ref{Well1}.}
\label{Well2}
\end{figure}

On the other hand, a somewhat more careful analysis of the effective potential ($v_{\rm eff}(s)$) function
(\ref{EffecPot}) suggests that by increasing the value of $C$, for a given value of $m$, the attractive contribution
may be greater than the repulsive one (this is shown in Figure \ref{Well3} for $C>1$).
\vspace{1cm}
\begin{figure}[H]
\begin{center}
\scalebox{0.85} 
{
\begin{pspicture}
{\includegraphics{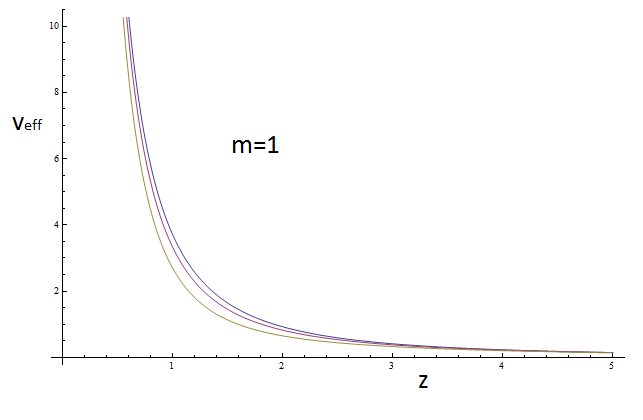}}
\end{pspicture} 
}
\end{center}
\caption{Effective potential for the case $m=1$ and the value of $C$ greater than the previous ones for the situation $C > 1$.}
\label{Well3}
\end{figure}

Next we obtain the upper limit number of two-quasiparticle bound states for any value of the angular momentum,
$m$. We start with the angular momentum $m > 0$ and use the following 

\begin{theorem}[N. Set\^o~\cite{Seto}, Theorem 3.2]
For $2m+d-2 \geqslant 1$, the number of bound states, $N_d^m$, produced by the potential $V$, satisfying 
the condition
\[
\int_0^\infty dr\,\,r\,|V(r)| < \infty\,\,,
\]
in the $m$-th wave in the $d$-dimensional space satisfies the inequality
\begin{align}
N_d^m < \frac{1}{2m+d-2} \int_0^\infty dr\,\,r\,|V(r)|\,\,.
\label{TheoSeto1}
\end{align}
\label{TheoSeto1a}
\end{theorem}

First of all, we have to be careful in the above estimate calculation due to dimensional consistency condition 
fulfillment, therefore we shall carry out the following substitution:
\[
\int_0^\infty dr\,\,r |V(r)| \longrightarrow \frac{2 \mu}{\hbar^2} \int_0^\infty dr\,\,r\,|V(r)|\,\,.
\]
Thus, we can evaluate the expression (\ref{TheoSeto1}) appropriately for $d=2$ and $m>0$
(non vanishing angular momentum), obtaining:
\[
N_2^m < \frac{1}{2m} \frac{2 \mu}{\hbar^2} \int_0^\infty dr\,\,r\,|V(r)|
=\frac{1}{2m} \frac{2 \mu \alpha}{\hbar^2 \beta^2} \underbrace{\int_0^\infty ds\,\,s\,K_0(s)}_{1}
=\frac{C}{2m}\,\,.
\]
The solution of the last integral is displayed in \cite[{\bf 6.561}, $16.$, p.676]{GR}, moreover, by analyzing the result above we conclude that
if $C/2m <1$, which implies that $C/2 < m$, there will be no bound states at all, as already conjectured through graphical
analysis (see Figure \ref{Well3}) -- for a fixed value of $C$, provided $m \leqslant C/2$, the greater the angular momentum ($m$) the
less the number of two-quantum bound states.

Let us now consider the vanishing angular momentum case, where $m=0$. It shall be pointed out that for zero angular
momentum the Theorem \ref{TheoSeto1a} does not apply. Therefore, we have to resort to the following

\begin{theorem}[N. Set\^o~\cite{Seto}, Theorem 5.1]
The number of bound states, $N_2^0$, produced by the potential $V$, satisfying 
the condition
\begin{align}
\int_0^\infty dr\,\,r \left(1+\left|\ln \frac{r}{R}\right|\right)|V(r)| < \infty\,\,,
\label{TheoSeto2}
\end{align}
in the $0$-th wave state in the two-dimensional space satisfies the inequality
\begin{align}
N_2^0 < 1+\frac{\displaystyle{\frac{1}{2}\int_0^\infty dr\,\,r\,|V(r)| \left(\int_0^\infty ds\,\,
s\,\left|\ln \frac{r}{s}\right|\,|V(s)|\right)}}{\displaystyle{\int_0^\infty dr\,\,r\,|V(r)|}}\,\,.
\label{TheoSeto2a}
\end{align}
\label{TheoSeto2b}
\end{theorem}

We shall verify if the potential $K_0$ satisfies the condition (\ref{TheoSeto2}). For this, again, based on
dimensional consistency, we make the following substitution:
\[
\int_0^\infty dr\,\,r \left(1+\left|\ln \frac{r}{R}\right|\right)|V(r)| \longrightarrow 
C \int_0^\infty ds\,\,s\,(1+|\ln s|)K_0(s)\,\,.
\] 
Now we rewrite the last integral as follows:
\begin{align*}
\int_0^\infty ds\,\,s\,(1+|\ln s|)K_0(s)
=\int_0^\infty ds\,\,s\,K_0(s)-\int_0^1 ds\,\,s \ln s\,K_0(s)
+\int_1^\infty ds\,\,s \ln s\,K_0(s)\,\,.
\end{align*}
Let us analyze each of the three integrals above separately:
\begin{enumerate}
\item The first integral is presented in \cite[{\bf 6.561}, $16.$, p.676]{GR}, and it is equal to $1$.

\item In the case of the third integral, we consider that
\[
\int_1^\infty ds\,\,s\,\ln s\,K_0(s) < \int_0^\infty ds\,\,s^2\,K_0(s)=\frac{\pi}{2}\,\,,
\]
since for $s > 1$, we have $\ln s < s$. The integral above at the right-hand side was also obtained from 
Ref.~\cite[{\bf 6.561}, $16.$, p.676]{GR}.

\item Finally, let us evaluate the second integral. Note that the function $f(s)=s \ln s\,K_0(s)$, defined 
in the open interval $(0,1)$ is well behaved in this interval, in the sense that there are no singularities
of any kind. For $s=1$, $K_0(1) \simeq 0,4210244382$ and therefore $f(1)=0$ due to the logarithmic
term. In order to verify the behavior of $f$ when $s \to 0$, we will need the following fact:
\[
K_0(s) \simeq -\ln \frac{s}{2}
\quad \text{as} \quad s \to 0\,\,.
\]
Then,
\[
\lim_{s \to 0^+} f(s) \simeq \lim_{s \to 0^+} s \ln s \ln \frac{s}{2}=0\,\,.
\]

Since all integrals are finite, it is proved that the potential $V(r)=-C K_0(r)$ respects the condition
(\ref{TheoSeto2}).
\end{enumerate}

Now, we have the endorsement to determine the limit for the number of bound states for the case when $m=0$.
Again, moving to a dimensionally consistent form, we rewrite the second piece of (\ref{TheoSeto2a}):
\[
C \left[\frac{\displaystyle{\frac{1}{2}\int_0^\infty dr\,\,r \,K_0(r) \left(\int_0^\infty ds\,\,
s\,\left|\ln \frac{r}{s}\right|\,K_0(s)\right)}}{\displaystyle{\int_0^\infty dr\,\,r\,K_0(r)}}\right]\,\,.
\]

From the integration in the variable $s$, we obtain
\[
\int_0^\infty ds\,\,s\,\left|\ln \frac{r}{s}\right|\,K_0(s)
=r^2 \int_0^\infty dt\,\,t\,|\ln t|\,K_0(rt)\,\,,
\]
where $t=s/r$. The latter is separated into
\[
r^2 \left(\int_1^\infty dt\,\,t\,\ln t\,K_0(rt)-\int_0^1 dt\,\,t\,\ln t\,K_0(rt)\right)
=r^2 \bigl({\rm I}_1+{\rm I}_2\bigr)\,\,.
\]
Using the Mathematica package we obtain
\begin{align}
{\rm I}_1=\int_1^\infty dt\,\,t\,\ln t\,K_0(rt)=\frac{K_0(r)}{r^2}\,\,,
\label{Int1}
\end{align}
and
\begin{align}
{\rm I}_2=\int_0^1 dt\,\,t\,\ln t\,K_0(rt)=-\frac{\gamma+K_0(r)+\ln \frac{r}{2}}{r^2}\,\,.
\label{Int2}
\end{align}
where $\gamma$ is the constant of Euler-Mascheroni. Hence, it follows that
\[
\int_0^\infty ds\,\,s\,\left|\ln \frac{r}{s}\right|\,K_0(s)=\gamma+2  K_0(r)+\ln \frac{r}{2}\,\,.
\] 

By integrating the variable $r$, we obtain
\begin{align*}
\int_0^\infty dr\,\,r\,K_0(r) \Bigl(\gamma+2  K_0(r)+\ln \frac{r}{2}\Bigr)
&=\gamma \int_0^\infty dr\,\,r\,K_0(r)+2 \int_0^\infty dr\,\,r\,K_0^2(r) \\[3mm]
&\qquad+\int_0^\infty dr\,\,r\,K_0(r) \ln \frac{r}{2}\,\,.
\end{align*}
From Ref.~\cite[p.665 and p.676]{GR} we get that the sum of the first two terms is $\gamma+1$
and the substitution $r=2t$ in the last integral, together with the expressions (\ref{Int1}) and (\ref{Int2}), shows
us that the last integral is $-\gamma$~\cite{NRG}.

Consequently, putting together all the above analysis we find that
\[
C \left[\frac{\displaystyle{\frac{1}{2}\int_0^\infty dr\,\,r\,K_0(r) \left(\int_0^\infty ds\,\,
s\,\left|\ln \frac{r}{s}\right|\,K_0(s)\right)}}{\displaystyle{\int_0^\infty dr\,\,r\,K_0(r)}}\right]=\frac{C}{2}
\quad \Longrightarrow \quad
N_2^0 < 1+ \frac{C}{2}\,\,.
\]

In short, we have
\begin{align} 
N_2^m <
\begin{cases}
\displaystyle{1+\frac{C}{2}} \quad &\text{for} \quad m=0 \\[5mm]
\displaystyle{\frac{C}{2m}} \quad &\text{for} \quad m \geqslant 1
\end{cases}\,\,.
\label{NumBS}
\end{align}

\section{\bf Proof of Proposition \ref{Specprop}}
\label{Sec6}
\hspace*{\parindent}
We just need to prove that there exists no other eigenvalue below $\inf\,\sigma_{\rm ess}(H)$.
The uniqueness of the ground state follows by standard arguments~\cite{RSIV,LM}. Indeed,
the potential $V(r)=-\alpha K_0(\beta r)$ satisfies the assumptions given in Ref.~\cite[Theorem 11.8]{LM}
which ensure that there is a unique minimizer, $\Psi_0$, up to a constant factor. On the other hand, according
to Eq.(\ref{NumBS}) the constant $C$ has a direct effect on the number of bound states and for $0 < C < 2$
if $m=0$, or $C=2m$ if $m \geqslant 1$, there will be exactly a single bound state, with a simple eigenvalue.
This will be found only in the $s$-wave state, {\em i.e.}, for vanishing angular momentum, $m=0$\footnote{See Ref.\cite{Chadan}
on the discussion about $s$-wave bound states for free particle in planar systems.}. Since the spectrum of the Schr\"odinger
operator (\ref{HOper}) starts at $-C\alpha/4\pi$, we conclude that $H$ has at most one isolated eigenvalue (it has exactly one).
This proves the proposition.

\begin{remark}
Proposition \ref{Specprop} is important in the case of high-$T_{\rm c}$ superconductors because it establishes an estimate
for the ground state energy $E_0=\inf\,\sigma(H)$, where $H$ is the Schr\"odinger operator (\ref{HOper}). In other words,
in case of zero angular momentum, it is possible to excite the two-quasiparticle pair by adding kinetic energy to it. Hence,
one obtains for the energy gap the relation $\Delta=C\alpha/4\pi$.
\end{remark}

\section{\bf Summary and perspectives}
\label{Sec7}
\hspace*{\parindent}
In this work we have demonstrated the self-adjointness of the non-relativistic hamiltonian 
operator for any three space-time dimensional quantum electrodynamics model (QED$_3$) 
exhibiting an attractive scattering potential of the type $V(r)=-\alpha K_0(\beta r)$. We also obtained information
about the discrete and essential spectra of the Schr\"odinger operator modeling the non-relativistic approximation
of the model. We have proved the existence of two-quantum bound states and computed the upper limit number
of these bound states for any value of the angular momentum, $m$. In addition to that, we obtain
an explicit estimate for the energy gap in case of zero angular momentum. This result seems to be relevant for either $s$-wave Cooper pairing-type high-$T_c$ superconductors~\cite{high-Tc,christiansen-delcima-ferreira-helayel} or $s$-wave electron-polaron--electron-polaron bound states, the $s$-wave bipolarons, in mass-gap graphene systems~\cite{graphene,graphene1,graphene2}. For future investigations, we aim: $(i)$ pursue possible applications to two-dimensional materials 
such as high-$T_c$ superconductors, graphene and topological insulators; $(ii)$ to study 
computationally and analytically the dynamics and thermodynamics of two-dimensional fermion
gas interacting via the scattering potential $V(r)=-\alpha K_0(\beta r)$, as well as verify possible 
phase transitions and compute critical parameters. On the other hand, from the theoretical point of 
view, the stability of such two-quasiparticle bound states is an issue to be analyzed, including: 
$(iii)$ the presence of applied magnetic fields and their interaction with the quanta spin;
$(iv)$ the relativistic kinematics, where $\sqrt{-\Delta}$ replaces $-\Delta$ in the kinetic energy;
$(v)$ models with kinetic energy described by the Dirac operator. All of these issues are in progress.

\section*{\bf Acknowledgements}
\hspace*{\parindent}
The authors would like to thank Jakson M. Fonseca for important remarks on the dimensionality
of the wave function in quantum mechanics, thanks are also due to the anonymous referee for very pertinent comments and suggestions.  B.C. Neves was supported by the Conselho Nacional
de Desenvolvimento Cient\'\i fico e Tecnol\'ogico (CNPq). 

\section*{Author's Contributions}
\hspace*{\parindent}
All authors contributed equally to this work.

\section*{Data Availability}
\hspace*{\parindent}
The data that support the findings of this study are available from the corresponding author upon reasonable request.




\begin{thebibliography}{99}
\bibitem{FT} We are adopting the following convention for the Fourier transform:
\begin{align*}
\bigl[{\mathscr F}\Psi\bigr](\boldsymbol{k})&=\widehat{\Psi}(\boldsymbol{k})
=\int_{\oR^n} d^n\boldsymbol{x}\,\,e^{i \hbar^{-1}\boldsymbol{k} \cdot \boldsymbol{x}}\,{\Psi}(\boldsymbol{x})\,\,, \\[3mm]
\bigl[{\mathscr F}^{-1}\widehat{\Psi}\bigr](\boldsymbol{x})&=\Psi(\boldsymbol{x})
=\frac{1}{(2\pi \hbar)^n}\int_{\oR^n} d^n\boldsymbol{k}\,\,e^{-i \hbar^{-1}\boldsymbol{k} \cdot \boldsymbol{x}}\,\widehat{\Psi}(\boldsymbol{k})\,\,.
\end{align*}
\bibitem{NRG} At this point it is important to note that the expression
\[
\gamma=-\int_0^\infty dr\,\,r\,K_0(r) \ln \frac{r}{2}\,\,,
\]
is a new integral representation for the Euler-Mascheroni constant, so far not found in the literature.
\end{thebibliography}
\end{document}